\documentclass[a4paper, UKenglish, cleveref, nameinlink,colorlinks=true, linkcolor=blue!75!green, urlcolor=blue!75!green, citecolor=blue!75!green, autoref, thm-restate]{lipics-v2021}

\usepackage{paralist}
\usepackage{dsfont}

\usepackage[textsize=footnotesize]{todonotes}

\usepackage{apptools,thmtools}
\let\origrestatable=\restatable
\renewcommand{\restatable}[1][]{%
  \origrestatable[%
    \ifstrempty{#1}%
      {$\star$}%
      {\IfAppendix{\hyperref[#1]{$\star$}}{\hyperref[#1*]{$\star$}}}%
  ]%
}

\DeclareMathOperator{\skel}{skel}
\DeclareMathOperator{\tw}{tw}
\DeclareMathOperator{\bw}{bw}

\nolinenumbers
\Crefname{theorem}{Theorem}{Theorems}
\Crefname{theorem2}{Theorem}{Theorems}
\Crefname{lemma}{Lemma}{Lemmas}
\Crefname{lemma2}{Lemma}{Lemmas}
\Crefname{figure}{Fig.}{Figs.}
\Crefname{algorithm}{Listing}{Listings}
\Crefname{section}{Section}{Sections}
\Crefname{observation}{Observation}{Observations}
\Crefname{property}{Property}{Properties}
\Crefname{lemma}{Lemma}{Lemmas}
\Crefname{claim}{Claim}{Claims}
\Crefname{claimx}{Claim}{Claims}
\Crefname{figure}{Fig.}{Figs.}
\Crefname{subfigure}{Fig.}{Figs.} 
\Crefname{minipage}{Fig.}{Figs.}
\Crefname{enumi}{Condition}{Conditions}
\Crefname{reductionrule}{Rule}{Rules}

\usepackage{xcolor}
\definecolor{realblue}{rgb}{0,0,1}
\definecolor{defblue}{rgb}{0.274,0.392,0.666}
\definecolor{darkerblue}{rgb}{0.094,0.455,0.804}
\definecolor{darkblue}{rgb}{0.063,0.306,0.545}
\definecolor{linkblue}{rgb}{0.098,0.098,0.4392}
\definecolor{red}{rgb}{0.627,0.117,0.156}
\definecolor{green}{RGB}{51, 153, 102}
\definecolor{orange}{rgb}{0.903,0.739,0.382}
\definecolor{realred}{rgb}{1,0,0}
\definecolor{lipicsblue}{rgb}{0.08235294118,0.3098039216,0.537254902}

\usepackage{cite} %
\usepackage{xspace} %

\let\emph\relax
\DeclareTextFontCommand{\emph}{\color{defblue}\em}
\DeclareTextFontCommand{\bl}{\color{lipicsblue}}
\hypersetup{colorlinks=true,
    linkcolor=lipicsblue,
    anchorcolor=lipicsblue,
    citecolor=lipicsblue,
    filecolor=lipicsblue,
    menucolor=lipicsblue,
    urlcolor=lipicsblue,
    bookmarksopen=true,
    bookmarksopenlevel=2,
    bookmarksnumbered=true,
    plainpages=false,
    }

\newcommand{\NP}{\textsf{NP}\xspace}

\newcommand{\ifaces}[1]{\mathcal F^{\mathrm{i}}_{#1}}
\newcommand{\ofaces}[1]{\mathcal F^{\mathrm{o}}_{#1}}
\newcommand{\bfaces}[1]{\mathcal F^{\mathrm{b}}_{#1}}
\newcommand{\rfaces}[1]{{\mathcal F}_{#1}}
\newcommand{\iv}[1]{\mathcal V^{\mathrm{i}}_{#1}}
\newcommand{\ov}[1]{\mathcal V^{\mathrm{o}}_{#1}}
\newcommand{\bv}[1]{\mathcal V^{\mathrm{b}}_{#1}}
\newcommand{\rv}[1]{\mathcal V_{#1}}
\newcommand{\desc}[1]{\Omega(#1)}

\usepackage{graphicx} %

\hideLIPIcs
\newcommand{\titleproposal}{
Beyond Degree Four:\\ Near-Orthogonal Planar Drawings
}
\title{\titleproposal}
\titlerunning{Beyond Degree Four: Near-Orthogonal Planar Drawings}

\authorrunning{P.~Angelini, S.~Cornelsen, G.~{Da~Lozzo}, S.~Hong, and I.~Rutter}

\author{Patrizio Angelini}{John Cabot University, Rome, Italy}{pangelini@johncabot.edu}{https://orcid.org/0000-0002-7602-1524}{}
\author{Sabine Cornelsen}{University of Konstanz, Germany}{sabine.cornelsen@uni-konstanz.de}{https://orcid.org/0000-0002-1688-394X}{}
\author{Giordano Da Lozzo}{Roma Tre University, Italy}{giordano.dalozzo@gmail.com}{https://orcid.org/0000-0003-2396-5174}{}
\author{Seok-Hee Hong}{University of Sydney, Australia }{seokhee.hong@sydney.edu.au}{https://orcid.org/0000-0003-1698-3868}{}
\author{Ignaz Rutter}{University of Passau, Germany}{rutter@fim.uni-passau.de}{https://orcid.org/0000-0002-3794-4406}{}

\Copyright{Patrizio Angelini, Sabine Cornelsen, Giordano {Da~Lozzo}, Seokhee Hong, and Ignaz Rutter}

\ccsdesc[500]{Theory of computation~Fixed parameter tractability}
\ccsdesc[500]{Theory of computation~Computational geometry}
\ccsdesc[500]{Mathematics of computing~Graph algorithms}

\category{Track 1, Long Paper} 

\acknowledgements{This work started at the Bertinoro Workshop on Graph Drawing BWGD 2026.}

\keywords{planar orthogonal drawings, unbounded degree, FPT, PTAS, treewidth}%

\begin{document}

\maketitle

\begin{abstract}
Orthogonal planar drawings constitute a classical and mainstream research topic in graph drawing due to their clarity and wide applicability. In an orthogonal planar drawing of a graph, each face is represented as an {\em orthogonal polygon}, that is, a polygon whose edges 
are either horizontal or vertical. Yet a planar graph admits such a representation if and only if its maximum degree is at most~four. 

In this paper, we consider planar polyline drawings of graphs with unrestricted maximum degree. We focus on drawings that are {\em ``close to orthogonal''}, where closeness is measured by the number of faces that are not orthogonal polygons. We show that, even when the input graph is triconnected and thus has a unique planar embedding, the problem of testing whether there exists a planar polyline drawing with at most $h$ non-orthogonal faces is $\mathsf{NP}$-complete. Motivated by this computational hardness, we study parameterized and approximation algorithms. In the fixed-embedding setting, we prove that the problem admits linear-time FPT algorithms parameterized by (i) the outerplanarity index and (ii) the natural parameter $h$. 
In addition,  we provide an FPT algorithm parameterized by the treewidth and a polynomial-time approximation scheme. 
In the variable-embedding setting, we give an FPT algorithm parameterized by treewidth for biconnected graphs.

\subparagraph{Generative AI Declaration}
Generative AI was solely used to improve the writing quality of some paragraphs in the preparation of this article.

\end{abstract}

\section{Introduction}

    Orthogonal planar drawings are among the most classic and extensively studied paradigms in graph drawing. In this model, each vertex is represented by a point, while each edge is depicted as an \emph{orthogonal polyline}, i.e., a sequence of horizontal and vertical segments. This representation has proven particularly useful in domains such as circuit design, software engineering, database design, and of course, network visualization~\cite{BatiniTT84,9780470073049.ch3,EiglspergerGKKJLKMS04,Juenger04,lengauer:90}. This enduring popularity stems from the clear visual structure created by the combination of planarity and axis-aligned edge routes, which make connections much easier to follow and interpret than in more general planar drawing styles. 

    A large body of work has identified the number of bends as a key factor influencing readability and visual quality in orthogonal drawings. Experimental results by Bertolazzi, Di Battista, and Didimo~\cite{BertolazziBD00} indicate that minimizing bends has a positive impact on several quality measures of drawing effectiveness. The algorithmic study of orthogonal graph drawings has therefore, for decades, focused on bend-minimum and bend-bounded drawings and on the efficient computation of such drawings. In the variable embedding setting, Garg and Tamassia~\cite{GargT01} showed that testing whether a graph admits a \emph{rectilinear planar drawing}, i.e., an orthogonal planar drawing without bends, is \NP-complete, confirming a conjecture of Storer~\cite{Storer80}, and that approximating the minimum number of bends for an $n$-vertex graph 
    with an $O(n^{1-\epsilon})$ error is \NP-hard for any $\epsilon>0$. At the same time, several positive results are known for important restricted graph classes, such as planar graphs of maximum degree~$3$~\cite{BattistaLV98,ChangY17,DidimoLP18,DidimoLOP20}, series-parallel graphs~\cite{DidimoKLO23,ZhouN08},  and outerplanar graphs~\cite{Frati22}, and for the fixed embedding setting~\cite{Tamassia87}. Bend-minimum orthogonal drawings have also been investigated through the lens of parameterized complexity~\cite{DidimoL98,  GiacomoDLMO24,GiacomoLM22,JansenKKLMS23}.

    A well-known structural limitation of the orthogonal paradigm is that it naturally applies only to graphs of small degree. Indeed, orthogonal representations rely on the four axis-aligned ports around each vertex. In fact, already in~1981,  Valiant~\cite{Valiant81} proved that a graph has an orthogonal drawing if and only if it has maximum degree $4$. 
    Due to their sustained appeal, attempts at extending the orthogonal drawing style to higher-degree graphs have been pursued. A notable example of maximum degree-$8$ graphs, originating in the context of metro map visualization and map schematization~\cite{HongMN06,StottRMW11,Wolff13}, is octilinear drawings.  In an \emph{octilinear drawing}, every edge is drawn as a sequence of horizontal, vertical, and diagonal line segments with a slope of $\pm 45^\circ$. For planar graphs of unbounded degree, F{\"{o}}{\ss}meier and Kaufmann~\cite{FossmeierK95} proposed an extension of the orthogonal drawing paradigm via the \emph{Kandinsky drawing model}, where vertices are represented as axis-aligned boxes of uniform size and edges are depicted as orthogonal polylines attached to any side of such boxes. 
    A different approach, tailored for directed graphs, is the \emph{L-drawing standard}, introduced 
    in ~\cite{KariOALBDPRT18a}. L-drawings combine orthogonal representations with adjacency-matrix-like representations: each vertex is assigned a point with exclusive $x$- and $y$-coordinates, and each edge is drawn as a $1$-bend polyline consisting of one vertical segment leaving its tail and one horizontal segment entering its head. This model retains the strong axis-aligned regularity of orthogonal drawings while allowing multiple edges to overlap along portions of rows and columns. By allowing edges to share ports, unlike in the classic orthogonal setting, this model supports representations of directed graphs of arbitrary degree.
    For further research related to octilinear, Kandinsky, and L-drawings, see,  \cite{BekosG0015,BekosKK17,BekosFK19}, \cite{BlasiusBR14,Bekos00S15,BekosFK19}, and \cite{AngeliniCCL22,ChaplickCCLNPT023,AngeliniCCL24},~respectively.

    \subparagraph{Our contributions.} In this research, we follow a different, and in our view complementary, route to transfer to planar graphs of arbitrary degree the core advances of orthogonal drawings. Rather than extending orthogonality by permitting additional prescribed segment directions %
    (as in the octilinear model), allowing overlaps between different edges (as in the L-drawing model) or enlarging the geometric vocabulary for vertices (as in the Kandinsky model), we seek planar polyline drawings that remain orthogonal on most faces, while all deviations from orthogonality are confined to a small number of faces. This perspective is aimed at preserving the geometric discipline of orthogonal drawings over the majority of the embedding, while providing the flexibility required to accommodate vertices of arbitrarily large degree. 
    
    A face in a planar drawing of a graph is \emph{orthogonal}, if it is bounded by an orthogonal polygon, otherwise it is \emph{non-orthogonal}. In a planar orthogonal drawing, all faces are orthogonal. 
    Observe that, among the edges incident to a vertex, at most four can be drawn with a horizontal or vertical segment at the vertex. The faces incident to any further edge are therefore necessarily non-orthogonal. Given a connected planar graph $G$, we want to find a planar drawing of $G$ in which as few faces as possible are \emph{non-orthogonal}. We call this optimization problem \textsc{Non-Orthogonal Face Minimization}, 
    and we call \textsc{$h$-Non-Orthogonal Faces} the related decision problem of testing for the existence of a planar drawing with at most $h$ non-orthogonal faces. Our results are as follows:
    \begin{itemize}
    \item First, we show that \textsc{$h$-Non-Orthogonal Faces} is \NP-complete even for triconnected planar graphs (\cref{thm:hardness}),  which excludes the existence of a polynomial-time algorithm for the problem even in the fixed-embedding setting (unless P=\NP). 
    \item Second, in the fixed-embedding setting, we show that \textsc{Non-Orthogonal Face Minimization} for an $n$-vertex graph $G$ can be solved in $2^{O(w)}n + \tau$ time, provided that a sphere-cut decomposition of $G$ of width $O(w)$ can be computed in time $\tau$ (\cref{thm:fpt_sc}).  Since an optimal sphere-cut decomposition can be computed in time~$\tau \in O(n^3)$, this yields a cubic-time FPT algorithm with respect to the treewidth $\tw(G)$ of $G$ (\cref{cor:fpt_tw}).  For series-parallel graphs, this yields a linear-time algorithm, since a sphere-cut decomposition of width~$2$ can be computed in linear time (\cref{cor:series-parallel}).  For $k$-outerplanar graphs, a sphere-cut decomposition of width~$O(k)$ can be computed in $O(kn)$ time, which yields a linear-time FPT algorithm with respect to the outerplanarity index (\cref{cor:fpt_outerplanar}).
    We use the latter result to derive a linear-time FPT algorithm with respect to the natural parameter $h$ (\cref{thm:fpt_natural}) and a polynomial-time approximation scheme (\cref{thm:ptas}).    
    \item Finally, in the variable embedding setting, we provide 
    an FPT algorithm parameterized by treewidth for biconnected planar graphs (\cref{thm:fpt_tw_variable}).
    \end{itemize}

\noindent

For statements marked with a $(\star)$, we provide the proof in the appendix.

\section{Preliminaries}\label{se:pre}

In this section, we give preliminaries and basic definitions. For standard concepts and terminology about graphs and their drawings, we refer the reader to~\cite{BattistaETT99,0030488}.  Unless stated otherwise, all graphs in this paper are connected, finite, and simple.

\subparagraph{Basic definitions.}
For a graph $G$, we denote by $V(G)$ the set of vertices of~$G$ and by $E(G)$ the set of edges of~$G$. 
A \emph{plane graph} is a planar graph equipped with a fixed planar embedding, which determines the outer face and the rotation scheme at each vertex.
For a plane graph~$G$, an \emph{outer layering} is defined as follows:  Let $L_i = \emptyset$ for all $i\leq 0$. For $i > 0$, define \emph{layer} $L_i$ to be the set of vertices incident to the outer face of $G - \bigcup_{j<i}L_j$.  
By construction, every vertex belongs to exactly one layer. 
Observe that the boundary of any face of $G$ contains vertices from at most two consecutive layers.
The \emph{outerplanarity index} of $G$ is the smallest $k$ such that $L_i= \emptyset$ for all $i > k$. Equivalently, it is the number of non-empty layers in the outer layering of $G$. A graph with outerplanarity index $k$ is called \emph{$k$-outerplanar}.

A decision problem or an optimization problem, respectively, is fixed-parameter-tractable (FPT) parameterized by a parameter $k$, if each instance of size $n$ can be solved in time $f(k)\cdot n^{O(1)}$ where $f$ is a computable function that only depends on $k$.

\begin{figure}
    \centering
    \includegraphics[width=0.5\linewidth]{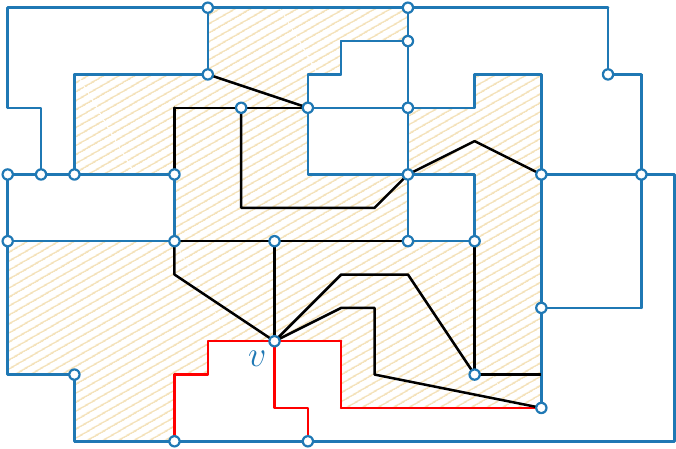}
    \caption{An $h$-near orthogonal planar drawing with $h=11$. The set $F$ of non-orthogonal faces is yellow and hatched. Edges in $E(F)$ are black. The three residual edges incident to vertex $v$~are~red.
    }
    \label{fig:h-near}
\end{figure}

\subparagraph{Near-orthogonal drawings.}
A planar drawing of a graph is \emph{$h$-near orthogonal} if it contains at most $h$ non-orthogonal faces, and a graph admitting such a drawing is \emph{$h$-near orthogonal}; refer to \cref{fig:h-near}. Let $F$ be the set of %
non-orthogonal faces of a planar drawing. 
We denote by $E(F)$ the set of edges that are only incident to faces in $F$.
Observe that, 
if an edge $e$ is not in $E(F)$, %
then it is drawn as an orthogonal polyline, as otherwise the face(s) incident to $e$ and not in $F$ would not be orthogonal. On the other hand, if $e$ is 
in $E(F)$, i.e., if $e$ is only incident to faces in $F$,  
then its drawing does not have to be an orthogonal polyline.
Based on the considerations above, for a set $F'$ of faces of a plane graph $G$, 
we call an edge a \emph{residual edge} w.r.t.\ $F'$ if it is not in $E(F')$. See \cref{fig:h-near}.  The \emph{residual degree} of a vertex~$v$ w.r.t.\ $F'$ is the number of residual edges that are incident to $v$.  Observe that the residual degree of each vertex w.r.t.\ the set of non-orthogonal faces is at most four. Based on these observations, we characterize the geometric problem from a combinatorial perspective.

\begin{restatable}[lem:char]{lemma}{characterization}
\label{lem:char}
    A planar (plane) graph $G$ is $h$-near orthogonal if and only if there is a set $F'$ of at most $h$ faces of $G$ in some (the given) embedding of $G$  such that each vertex has residual degree at most four~w.r.t. $F'$.
\end{restatable}

We can now formulate our optimization problem and its decision version as follows.
\medskip
\begin{description}
\item[\textsc{Non-Orthogonal Face Minimization}:]
Given a planar (plane) graph $G$, find a smallest set~$F'$ of faces 
of some (respectively, the given) embedding 
of $G$ such that each vertex has residual degree at most four w.r.t.~$F'$.
\item[\textsc{$h$-Non-Orthogonal Faces}:] 
Given a planar (plane) graph $G$, is there a set $F'$ of faces 
of some (respectively, the given) embedding 
of $G$ such that each vertex has residual degree at most four w.r.t.~$F'$ and $|F'| \leq h$?
\end{description}

We also consider the following generalization of
\textsc{Non-Orthogonal Face Minimization}. 

\begin{description}
\item[\textsc{Non-Orthogonal Face Minimization$^*$}:]
Given a planar (plane) graph $G$ and a subset $V' \subseteq V(G)$, find a smallest set~$F'$ of faces of some (respectively, the given) embedding 
of $G$ such that each vertex of $V'$ has residual degree at most four w.r.t.~$F'$.
\end{description}

\subparagraph{Tree Decomposition.}

A \emph{tree decomposition}\label{def:treeDecomposition} \cite{robertson/seymour3:84} of a graph $G$ is a tree $\mathcal T$ such that each node $\nu \in V(\mathcal T)$ is associated with a set $V(\nu) \subseteq V(G)$, the \emph{bag} of $\nu$, 
such that %
\begin{inparaenum}[(a)]
	\item $\bigcup_{\nu\in V(\mathcal T)}V(\nu)=V(G)$,
	\item for each edge $e$ in $E(G)$, there is a node  $\nu \in V(\mathcal T)$ such that $e \subseteq  V(\nu)$, and,
	\item for each $v\in V(G)$, the subgraph of $\mathcal T$ induced by the vertices $\{\nu\in V(\mathcal T) \; | \; v \in V(\nu)\}$ is connected.
\end{inparaenum}
The \emph{width} of a tree decomposition is $\max_{\nu \in V(\mathcal T)}|V(\nu)|-1$. 
The \emph{treewidth} $\tw(G)$ 
of a graph~$G$ is the minimum among the widths of any of its tree decompositions. 

\begin{figure}
    \centering
    \includegraphics[width=0.5\linewidth,page=2]{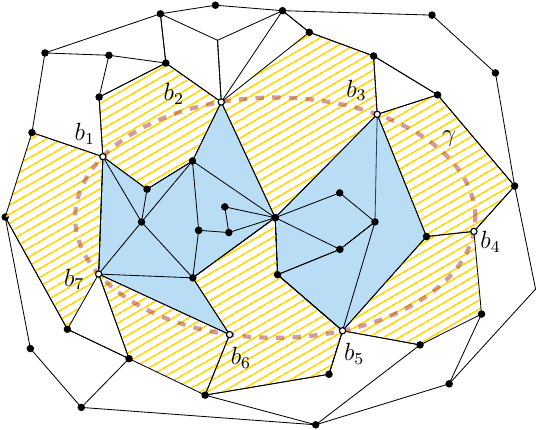}
    \caption{The partition of~$\mathcal F$ with respect to the noose $\gamma$ in interior, boundary, and outside faces (shaded blue, yellow, and white, respectively). The boundary vertices $b_1,\dots,b_6$ are filled with white.}
    \label{fig:noose}
\end{figure}

\subparagraph{Sphere-cut decomposition.} The concept of sphere-cut decomposition was introduced by Dorn, Penninkx, Bodlaender, and Fomin~\cite{DornPBF10}. 
Given a connected plane graph $G$ with planar drawing $\Gamma$ in the sphere, 
a \emph{sphere-cut decomposition} divides $G$ along simple closed curves, called \emph{geometric nooses}, that intersect $G$ only at vertices; refer to \cref{fig:noose}.  The decomposition yields a special \emph{branch decomposition}, i.e., a ternary tree $\mathcal T$ (that is, every internal node of~$\mathcal T$ has degree $3$) whose leaves correspond to the edges of $G$ and whose edges correspond to partitions of $E(G)$.   More precisely, if $e$ is an edge of  $\mathcal T$, let $E_1(e)$ and $E_2(e)$ be the set of edges of $G$ stored in the leaves of the two connected components of $\mathcal T - e$ and let $V_e$ be the \emph{boundary set} of $e$, 
that is, the set of vertices of $G$ that are incident to some edge in $E_1(e)$ and to some edge in $E_2(e)$. Then, there exists a geometric noose $\gamma$ that traverses only the vertices in $V_e$ and that separates~$E_1(e)$ and~$E_2(e)$ in the sense that they lie on different~sides~of~$\gamma$. The edge $e$ of $\cal T$ stores the set $V_e$ and the order in which such a set is traversed by $\gamma$.

In the original definition in~\cite{DornPBF10}, each noose is additionally required to traverse each face only once.  In this paper, we adopt a more combinatorial definition introduced by Dorn~\cite{dorn:stacs10}, which avoids such a requirement. 
A \emph{combinatorial noose} is a sequence $[v_0,f_0,v_1,f_1,\dots,f_{k-1},v_k,f_k]$   in the plane embedding $\cal E$ of a plane graph $G$ is an alternating sequence of vertices of $G$ and faces of $\cal E$ such that:%
\begin{itemize}
\item\label{item:csc:distinctVertices} $v_0 = v_k$ and $v_1,\dots,v_k$ are distinct;
\item for $i=0,\dots,k-1$, $f_i$ is a face incident to both $v_i$ and $v_{i+1}$; and
\item\label{item:csc:nested} if $f_i = f_j$ for any $i\neq j$ and $i,j = 0,\dots,k-1$, then the vertices $v_i$, $v_{i+1}$, $v_j$, and $v_{j+1}$ do not appear in the order $v_i$, $v_j$, $v_{i+1}$, $v_{j+1}$ on the boundary of the face $f_i = f_j$.
\end{itemize}
Observe that a geometric noose defines a combinatorial noose by listing the vertices and faces in the order it traverses them. Conversely, the properties of a combinatorial noose guarantee that there is a geometric noose visiting the same vertices and faces in the same order. 

We root~$\cal T$ by subdividing an arbitrary edge~$h$ of~$\cal T$ with a new node~$\rho$, which we select as the root.  Each node~$\mu \ne \rho$ of~$\cal T$ is associated with the (combinatorial) noose $\gamma$ of its parent edge, and it partitions the vertices of $G$ into three sets: the \emph{boundary vertices}~$\bv{\mu}$ that are traversed by~$\gamma$; the \emph{interior vertices} $\iv{\mu}$, for which all incident edges are leaves of the subtree of $\cal T$ rooted at $\mu$; and the \emph{outside vertices} $\ov{\mu}$, for which no incident edge is a leaf of this subtree.
Likewise, the set~$\mathcal F$ of faces of~$\cal E$ is partitioned into three sets (see \cref{fig:noose}): $\bfaces{\mu}$, which are traversed by~$\gamma$;~$\ifaces{\mu}$, for which all incident edges are leaves of the subtree of~$\cal T$ rooted at~$\mu$; and $\ofaces{\mu}$, for which no incident edge belongs to the subtree of~$\cal T$ rooted at~$\mu$.  \Cref{fig:sc-merge} shows the boundary and inner faces as well as the nooses corresponding to a node~$\mu$ and its two children~$\nu,\lambda$.

The \emph{width} of a sphere-cut decomposition of $G$ is $\max_{e \in E(\mathcal T)} V_e$.  The minimum among the widths of any sphere-cut decomposition of $G$ is denoted by $\bw(G)$.   
Robertson and Seymour~\cite{DBLP:journals/jct/RobertsonS91} show that $\bw(G) \in \Theta(\tw(G))$, more precisely, 
\mbox{$\bw(G)-1 \leq \tw(G) \leq \lfloor \frac{3}{2}\bw(G)\rfloor-1$.}
By triangulating the graph $G$ without asymptotically increasing the treewidth~\cite{BiedlV13}, and hence $\bw(G)$, and by applying the construction in~\cite{DornPBF10}, we obtain the following decomposition tool. 

\begin{restatable}[th:sc-construction]{proposition}{scgeneral}
\label{th:sc-construction}
    For a connected plane graph $G$, a sphere-cut decomposition of width $O(\bw(G))$ can be constructed in $O(n^3)$ time.
\end{restatable}

A decomposition whose width is linear in the outerplanarity index can~be~computed~faster.

\begin{proposition}\label{prop:sc-kouterplane}
    A sphere-cut decomposition of width $2k$  of a $k$-outerplane graph can be computed in $O(kn)$ time.
\end{proposition}
\begin{proof}
    Let $\mathcal F(G)$ be the set of faces of a connected plane graph~$G$. The \emph{vertex-face incident graph} $H$ of $G$ is a bipartite graph with vertex set $V(G) \cup \mathcal F(G)$ and an edge between a vertex $v \in V(G)$ and a face~$f \in \mathcal F(G)$ for each incidence of $v$ with~$f$. Let $T$ be a BFS-tree of $H$ rooted at the outer face of $G$ and let $h(T)$ be the height of $T$.
Then~$G$ has a sphere-cut decomposition of width at most $h(T)$~\cite{tamaki:esa03}  and such a sphere-cut decomposition can be computed in $O(h(T) \cdot |V(G)|)$ time~\cite{dorn:arxiv09,tamaki:esa03}. Let $L_i$, $i\geq 0$ be the layers
of $T$. Then $L_{2i-1}$, $i \geq 0$ is an outer layering of $G$. Thus, if $G$ is $k$-outerplane then $h(T) \leq 2k$. 
\end{proof}

\subparagraph{SPQR-Trees.} 
An \emph{SPQR-tree} \cite{dt-olpt-96} $\mathcal T$ is a tree that stores the decomposition of a biconnected graph into its 3-connected components. The leaves of $\mathcal T$ are labeled $Q$ and the inner nodes are labeled $S$, $P$, or $R$ such that no two adjacent nodes of $\mathcal T$ have the same label. 
Each node~$\nu$ of $\mathcal T$ is associated with a graph, called its \emph{skeleton} skel$(\nu)$, with two designated vertices, called its \emph{poles}. If $\nu$ is an inner node of $\mathcal T$ of degree $d$, then skel$(\nu)$ has $d$ edges, called \emph{virtual edges}, which are in one-to-one correspondence with the neighbors of $\nu$. Moreover, the end-vertices of each virtual edge of skel$(\nu)$ are mapped to the poles of the neighbor of $\nu$ corresponding to this virtual edge. 
The skeleton of $\nu$ is a simple cycle, if $\nu$ is an $S$-node, $d$ parallel edges, if $\nu$ is a $P$-node, and a 3-connected graph, if $\nu$ is an $R$-node. The skeleton of a $Q$-node $\nu$ consists of two parallel edges: a \emph{real} edge and a virtual edge corresponding to the only neighbor~of~$\nu$.   

Considering $\mathcal T$ rooted at an arbitrary $Q$-node $\rho$, 
each node $\nu \neq \rho$ of $\mathcal T$ represents a graph $G^+_\rho(\nu)$ consisting of a graph $G_\rho(\nu)$ plus the virtual edge parent$(\nu)$ of skel$(\nu)$ that corresponds to the parent of $\nu$.
If $\nu$ is a $Q$-node different from $\rho$, then $G^+_\rho(\nu) =$ skel$(\nu)$. Otherwise, let $\nu_i$, $i=1,\dots,\ell$, be the children of $\nu$ and let $e_i$ be the virtual edge of skel$(\nu)$ corresponding to $\nu_i$. Then $G^+_\rho(\nu)$ is constructed from skel$(\nu)$ and $G^+_\rho(\nu_i)$, $i=1,\dots,\ell$ by \emph{merging} $e_i$ and parent$(\nu_i)$, i.e., by replacing $e_i$ in skel$(\nu)$ with $G_\rho(\nu_i)$, identifying the end-vertices of $e_i$ with the poles of $\nu_i$ according to their mapping. 
The graph $G_\rho(\rho)$ is obtained from $G_\rho(\nu)$, where $\nu$ is the only child of $\rho$, by replacing parent$(\nu)$ with a real edge. We omit the index $\rho$ if the root is clear from the context.
For a graph $G$ there is a unique SPQR-tree such that $G=G_\rho(\rho)$ for any choice $\rho$ of the root and such an SPQR-tree can be computed in~linear~time~\cite{gutwengerMutzelMarks:gd00}. 

A biconnected graph is \emph{series-parallel} if its SPQR-tree does not contain $R$-nodes. %
By making this SPQR-tree ternary while respecting a fixed embedding, we obtain the following.

\begin{restatable}[prop:sp-sc]{proposition}{scseriesparallel}
\label{prop:sp-sc}
    For a biconnected plane series-parallel graph, a sphere-cut decomposition of width $2$ can be computed in linear time.
\end{restatable}

\section{Problem Complexity}\label{sec:complexity}

\begin{theorem}\label{thm:hardness}
\textsc{$h$-Non-Orthogonal Faces} is \NP-complete both in the variable and in the fixed embedding case, even if the input graph is 3-connected. 
\end{theorem}

\begin{figure}[tb]
\centering
    \begin{subfigure}[b]{0.45\textwidth}
      \centering
      \includegraphics[page=10]{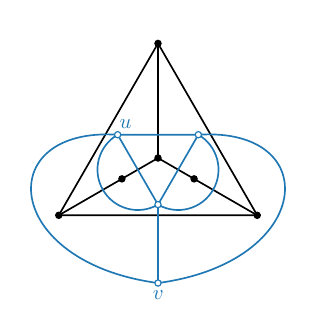}
	   \subcaption{}
        \label{fig:hardGraph}
    \end{subfigure}
    \begin{subfigure}[b]{0.45\textwidth}
      \centering
      \includegraphics[page=16]{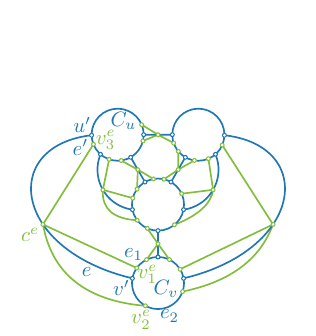}
	\subcaption{}
    \label{fig:hardGadget}
    \end{subfigure}
\\
    \begin{subfigure}[b]{0.45\textwidth}
      \centering
      \includegraphics[page=11]{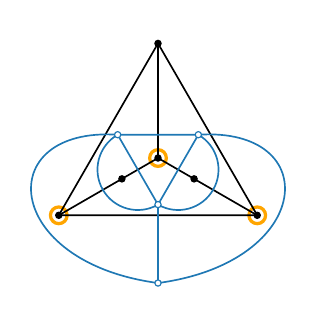}
	   \subcaption{}
       \label{fig:hardVC}
    \end{subfigure}
    \begin{subfigure}[b]{0.45\textwidth}
      \centering
      \includegraphics[page=18]{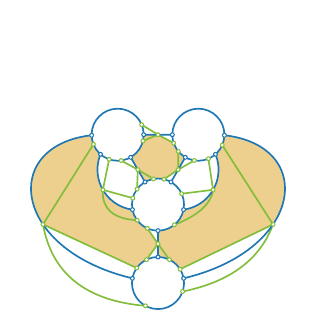}
	   \subcaption{}
        \label{fig:hardFC}
    \end{subfigure}
    \caption{Reduction from \textsc{Vertex Cover} to \textsc{$h$-Non-Orthogonal Faces}
    (a) Graph $H$ (black) and dual $H^*$ (blue).
    (b) Graph $G$, including $G'$ (in blue).
    (c) Vertex cover of $H$ of size $3$ (orange encircled).
    (d) Corresponds to a witness with $3+8$ faces (shaded).
    }\label{fig:hard}
\end{figure}

\begin{proof}
	Given a connected planar graph $G$ with a planar embedding and a set $F'$ of faces, it can be decided in polynomial time whether $|F'| \leq h$ and whether the residual degree of each vertex w.r.t. $F'$ is at most four. So, the problem is in NP.
	To prove NP-hardness, we reduce from vertex cover, which is NP-complete even if restricted to subdivisions of 3-connected planar graphs~\cite{mohar01}. See \cref{fig:hard}.  Let $H$ be a subdivision of a 3-connected planar graph. Let $H^*$ be the dual of $H$ w.r.t.\ the unique planar embedding of~$H$. 
	Observe that $H^*$ is 3-connected, but not necessarily simple.
	Replace each vertex in $H^*$ of degree $d$ by a cycle of length $d$. More precisely, let $v$ be a vertex of $H^*$ with incident edges $e_1,\dots,e_d$ in this order around~$v$. Replace $v$ by a cycle $C_v: v_1,\dots,v_d$ and for $i=1,\dots,d$, let $e_i$  be incident to $v_i$ instead of $v$. The resulting graph $G'$ is simple, 3-connected, and all vertices have degree~3. The edges of~$G'$ are either edges of $C_v$, $v \in V(H^*)$ or stem from $E(H^*)$. The faces are either bounded by some $C_v$,   $v \in V(H^*)$, or they are bounded by edges that stem from $E(H^*)$ and correspond to faces of $H^*$, i.e., to vertices of $H$. Let $f_v'$ be the face of~$G'$ that corresponds to the vertex~$v$ of $H$. 

    For each edge $e=\{u,v\}$ of $E(H^*)$, we construct a vertex of degree five. Tu this end we add four vertices: a vertex $c^e$ that subdivides $e$ and three vertices that subdivide edges of $C_v$ and $C_u$.  More precisely, let $v'$ and $u'$ be the end vertices of $e$ in $C_v$ and $C_u$,  let $e_1$ and $e_2$ be the two edges incident to $v'$ in $C_v$, and let $e'$ be an edge incident to $u'$ in $C_u$. Then we subdivide $e_1$ with a vertex $v^e_1$, the edge $e_2$ with a vertex $v^e_2$, and $e'$ with a vertex $v^e_3$. We add edges between $c^e$ and   $v^e_i$, $i=1,2,3$. The resulting graph $G$ is 3-connected: First, we subdivide two edges ($e$ and $e_1$) and add an edge between the two subdivision vertices. Then, we subdivide an edge ($e_2$ or $e'$) and add an edge from the subdivision vertex to a vertex ($c^e$) that was not an end vertex of the subdivided edge. Both are Barnette-Gr\"unbaum operations~\cite{bg-stc3p-69}, which maintain 3-connectivity.

    The only vertices of degree greater than three of $G$ are the central vertices $c^e$,  $e \in E(H^*)$. 
    Let $E_f$ be the set of edges bounding a face $f$. 
    Then there is exactly one face $f_v$ in $G$ that contains the vertices $c^e$, $e \in E_{f'_v}$. We call $f_v$ the face of $G$ corresponding to $v$.

	We show that $H$ has a vertex cover with $k$ vertices if and only if $G$ has a planar drawing with at most $k+E(H)$ non-orthogonal faces. 
	Assume first that $H$ has a vertex cover $V'$ with $k$ vertices. The $k$ faces $f_v$, $v \in V'$ together contain all vertices $c^e$,  $e \in E(H^*)$ on its boundary. For each vertex $c^e$, we choose one edge $e'$ incident to $c^e$ embedded such that it is incident to one of the faces of $f_v$, $v \in V'$. The edge $e'$ is incident to one other face. So at most $k + |E(H)|$ faces are incident to an edge that we remove in order to obtain a graph with maximum degree four.
	
	Consider now a drawing of $G$ in which $F'$ is the set of faces that are not orthogonal and assume $|F'| \leq k + |E(H)|$. Starting with an empty set, we construct a set $V''$ of size $k$ such that all faces in $F'' := \{f_v; v \in V''\}$ together contain all vertices $c^e$, $e \in E(H^*)$, on their boundary.  This implies that $V''$ is a vertex cover of $G$.   We now describe the construction of~$V''$: Each vertex $c^e$, $e \in E(H^*)$ is incident to at least one edge contained in the boundary of two faces in $F'$. Pick for each $c^e$ one such edge and let $E'$ be the resulting set of edges. For each edge $e \in E'$ there is at least one incident face in $F'$ that is not incident to any other $c_{e'}$, $e' \in E(H^*) \setminus \{e\}$. If there is only one such face then the other face incident to $e$ is $f_v$ for some $v \in V(H)$ and we add $v$ to $V''$.   
Otherwise, arbitrarily pick one of the two vertices $v \in V(H)$ such that $f_v$ is incident to $c^e$ and add $v$ to $V''$.  In any case, for each vertex $c^e$ there is an exclusive face for which we did not add a vertex to $V''$. Thus, $|V''| \leq |F'| - |E(H)| = k$.
\end{proof}

\section{Fixed-Parameter Tractability for Fixed Embedding}\label{section:fpt}

In this section, we consider the case where an embedding of $G$ is fixed. 
We show %
FPT algorithms parameterized by the treewidth and the solution size.
Moreover, for the special case of series-parallel graphs or fixed outerplanarity index we provide a linear-time algorithm.

\begin{theorem}\label{thm:fpt_sc}
For $n$-vertex plane graphs, \textsc{Non-Orthogonal Face Minimization$^*$} can be solved in time~$2^{O(w)}n + \tau$, provided that a sphere-cut decomposition of $G$ of width $w$ can be computed in time $\tau$.
\end{theorem}

\begin{proof}
Let~$(G,V')$ be an instance of \textsc{Non-Orthogonal Face Minimization}$^*$, where $G$ is a plane graph with a fixed planar embedding~$\mathcal E$ and~$V' \subseteq V$ are the vertices whose residual degree should be at most~$4$. We denote the set of faces of~$\mathcal E$ by~$\mathcal F$.  We describe an algorithm that uses dynamic programming along a sphere-cut decomposition of~$G$, rooted at a node $\rho$ as described in \Cref{se:pre}.

Consider a node $\mu$ of $\cal T$ and let $G_\mu$ be the subgraph of $G$ induced by the edges that are the leaves of the subtree of $\cal T$ rooted at $\mu$. 
When considering the subinstance corresponding to~$\mu$, the \emph{relevant faces} are~$\rfaces{\mu} = \ifaces{\mu} \cup \bfaces{\mu}$ and the \emph{relevant vertices} are $\rv{\mu} = \iv{\mu} \cup \bv{\mu}$, i.e., we disregard the outer faces and vertices, respectively.
Let~$X \subseteq \mathcal F$ be a set of faces.  The \emph{$\mu$-residual degree} of a vertex~$v$ of $G_\mu$ with respect to~$X$ is the number of edges of $G_\mu$ incident to~$v$ that are not incident to two faces in~$X$.  We denote it by~$d_{\mu,X}(v)$.  Observe that the $\mu$-residual degree of $v$ with respect to~$X$ depends only on~$X \cap \mathcal \rfaces_{\mu}$.

Let~$F \subseteq \mathcal F$  be a set of faces such that all vertices of $V'$ have residual degree at most~$4$ with respect to~$F$.  Clearly, $F_\mu = F \cap \rfaces{\mu}$ has the property that all vertices in~$\rv{\mu} \cap V'$ have residual degree at most~$4$ with respect to $F_\mu$.
We call a subset of~$\rfaces{\mu}$ with this property~\emph{$\mu$-feasible}.
Consider a $\mu$-feasible set~$X \subseteq \mathcal F_\mu$. If~$X$ has the property that~(i) $X \cap \bfaces{\mu} = F \cap \bfaces{\mu}$, and (ii) $d_{\mu,X}(b)=d_{\mu,F}(b)$ for all boundary vertices $b \in \bv{\mu} \cap V'$, then all vertices in $V'$ have residual degree at most~$4$ with respect to~$F' := (F \cap \ofaces{\mu}) \cup X$.
In other words, by replacing the set of faces of $F$ that lie in the interior or on the boundary of the noose of $\mu$ with the faces in $X$, we obtain a set $F' \subseteq \cal F$ such that all the vertices of $V'$ have residual degree at most $4$ with respect to $F'$.

Thus, all the relevant information regarding a~$\mu$-feasible subset $X$ 
of faces can be encoded by a \emph{descriptor} $\delta \in \{0,1\}^{|\bfaces{\mu}|} \times \{0,1,2,3,4\}^{|\bv{\mu} \cap V'|}$ that stores (a) one bit per boundary face, indicating whether that boundary face is contained in $X$, and (b) one number per boundary vertex in~$V'$ that encodes the~$\mu$-residual degree of the respective boundary vertex with respect to $X$.  We denote by~$\desc{\mu} := \{0,1\}^{|\bfaces{\mu}|} \times \{0,1,2,3,4\}^{|\bv{\mu} \cap V'|}$ the set of all possible descriptors.  For a set~$X \subseteq \rfaces{\mu}$, we denote its descriptor by~$\Delta_\mu(X)$.
Given two $\mu$-feasible sets~$X,X' \subseteq \rfaces{\mu}$ with the same descriptor $\Delta_\mu(X) = \Delta_\mu(X')$ but~$|X| < |X'|$, any feasible solution~$F$ with~$F \cap \rfaces{\mu} = X'$ 
with~$F \cap \rfaces{\mu} = X'$ 
can be improved by taking~$(F \setminus X') \cup X$; the fact that~$(F \setminus X') \cup X$ is a feasible solution follows from the previous paragraph. 

To find an optimal solution, we apply dynamic programming on a sphere-cut decomposition~$\cal T$ of $G$. 
Our goal is to compute, for each node~$\mu$, a table~$T_\mu$ with one entry $T[\delta]$ for each descriptor $\delta \in \desc{\mu}$ that contains the minimum size of a $\mu$-feasible set $X \subseteq \rfaces{\mu}$ with~$\Delta_\mu(X) = \delta$ and~$\infty$ if no such set exists.
\begin{figure}[tb]
\centering
      \includegraphics[page=8, width=.8\textwidth]{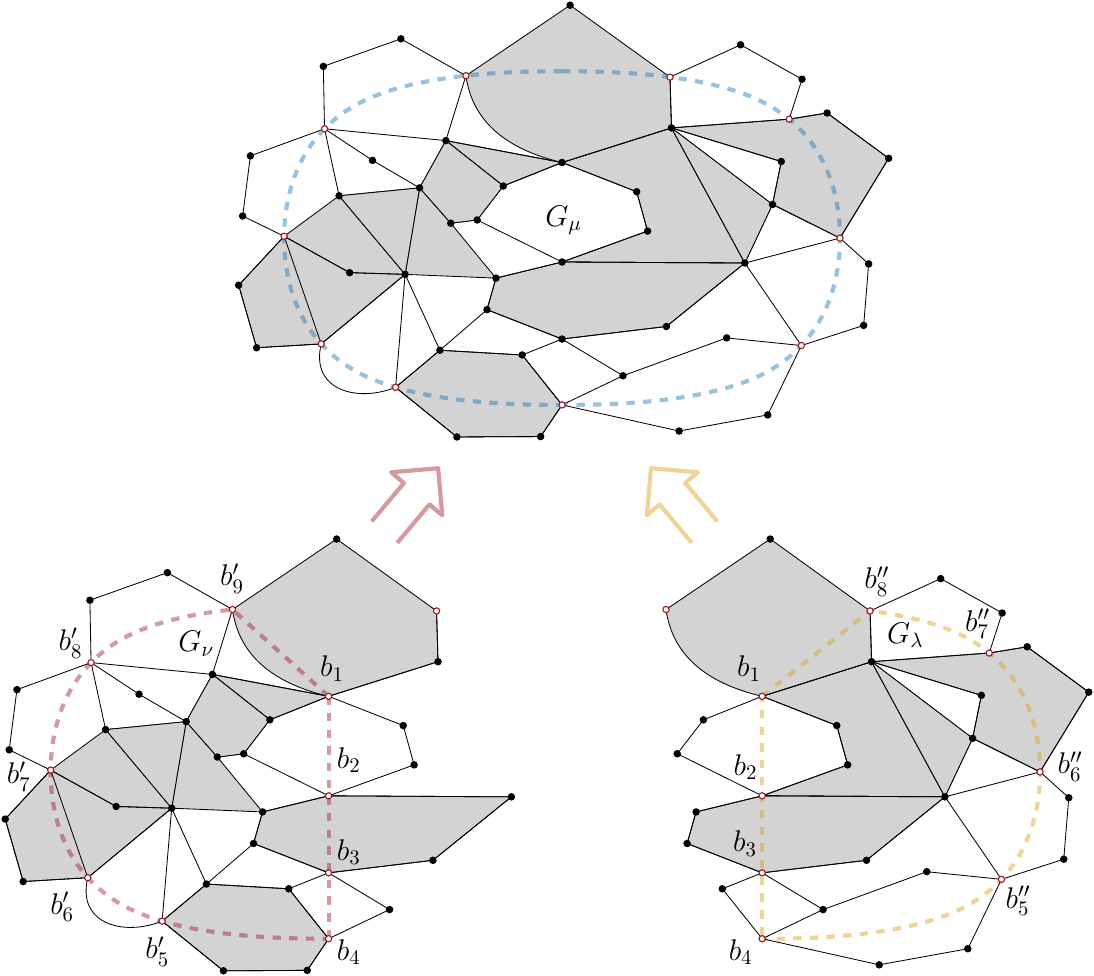}
 \caption{
 (top) Faces of $\rfaces{\mu}$ where the faces of a $\mu$-feasible set $X_\mu$ are shaded gray. (bottom left) Faces of $\rfaces{\nu}$ where the faces of $X_\nu = X_\mu \cap \rfaces{\nu}$ are shaded gray. (bottom right) Faces of $\rfaces{\lambda}$ where the faces of $X_\lambda = X_\mu \cap \rfaces{\lambda}$ are shaded gray.
 The descriptor of $\delta_\nu = \Delta_\nu(X_\nu)$ is $(0,1,0,1,0,1,0,0,1),(3,2,2,1,2,1,1,3,2)$, the descriptor of $\delta_\lambda = \Delta_\lambda(X_\lambda)$ is $(0,1,0,0,0,1,0,1),(1,1,2,1,2,2,1,1)$.
 }
\label{fig:sc-merge}
\end{figure}
In the base case, where~$G_\mu$ consists of a single edge~$e$, the table $T_\mu$ can be computed straightforwardly.  For example, if neither of the endpoints~$b_1,b_2$ of~$e$ are in~$V'$, then we have~$T_\mu[(0,0),(1,1)] = 0$, $T_\mu[(1,0),(1,1)] = T_\mu[(0,1),(1,1)] = 1$, and~$T_\mu[(1,1),(0,0)] = 2$ and all other entries of~$T_\mu$~are~$\infty$.  The other cases work similarly.

Now consider a node~$\mu$ of the sphere-cut decomposition $\cal T$ with children~$\nu$ and~$\lambda$; refer to \cref{fig:sc-merge}.   Note that~$\bv{\mu} \subseteq \bv{\nu} \cup \bv{\lambda}$.  
Let $X_\nu$ and $X_\lambda$ be a $\nu$-feasible and a $\lambda$-feasible subset of $\rfaces{\nu}$ and $\rfaces{\lambda}$,  respectively, such that
$X_\nu$ and $X_\lambda$ agree on their shared boundary faces, i.e., $X_\nu \cap (\bfaces{\nu} \cap \bfaces{\lambda}) = X_\lambda \cap (\bfaces{\nu} \cap \bfaces{\lambda})$.  Consider the set~$X_\mu = X_\nu \cup X_\lambda$.  Then~$X_\mu$ is~$\mu$-feasible if and only if $d_{G_\nu,X_\nu}(b) + d_{G_\lambda,X_\lambda}(b) \le 4$ for each shared boundary vertex $b \in \bv{\nu} \cap \bv{\lambda} \cap V'$.  And, since the only faces that are shared by~$X_\nu$ and~$X_\lambda$ belong to $\bfaces{\nu} \cap \bfaces{\lambda}$, we have $|X_\mu| = |X_\nu| + |X_\lambda| - |(X_\nu \cup X_\lambda) \cap (\bfaces{\nu} \cap \bfaces{\lambda})|$.  We note that the conditions necessary to ensure that~$X_\mu$ is $\mu$-feasible depend only on the descriptors~$\Delta_\nu(X_\nu)$ and~$\Delta_\lambda(X_\lambda)$ of~$X_\nu$ and~$X_\lambda$, respectively.  We call two descriptors~$\delta_\nu \in \desc{\nu}$ and~$\delta_\lambda \in \desc{\lambda}$ \emph{compatible}  if they agree on the shared boundary faces in~$\bfaces{\nu} \cap \bfaces{\lambda}$ and the sum of the two residual degrees in~$\delta_\nu$ and~$\delta_\lambda$ is at most $4$ for each shared boundary vertex in~$\bv{\nu} \cap \bv{\lambda} \cap V'$.  It is immediate from the definition  that the union of a $\nu$-feasible set of faces~$X_\nu$ and a~$\lambda$-feasible set of faces~$X_\lambda$ is~$\mu$-feasible if and only if $\delta_\nu$ and~$\delta_\lambda$ are compatible.  Moreover, the descriptor~$\Delta_\mu(X_\nu \cup X_\lambda)$ is unique and depends only on~$\delta_\nu$ and~$\delta_\lambda$, and we denote it by~$\delta_\nu \oplus \delta_\lambda$ in this case.
Finally, we note that the correction term~$|(X_\nu \cap X_\lambda) \cap (\bfaces{\nu} \cap \bfaces{\lambda})|$ is also determined by~$\delta_\nu$ and~$\delta_\lambda$, and we denote this number by~$s_\mu(\delta_\nu,\delta_\lambda)$.  
We have the following claim.

\begin{claim}
The entries of the table $T$ satisfy the following recurrence
\begin{equation}T_\mu[\delta_\mu] = \min_{\substack{\delta_\nu \in \desc{\nu}, \delta_\lambda \in \desc{\lambda}\\ \delta_\nu \oplus \delta_\lambda = \delta_\mu}} T_\nu[\delta_\nu] + T_\lambda[\delta_\lambda] - s_\mu(\delta_\nu,\delta_\lambda) \, .\label{eq:dp-table}\end{equation}
\end{claim}

\begin{claimproof}
We first show ``$\le$'':
Let~$\delta_\nu \in \desc{\nu}$, $\delta_\lambda \in \desc{\lambda}$ be such that~$\delta_\mu = \delta_\nu \oplus \delta_\lambda$.  We show~$T_\mu[\delta_\mu] \le T_\nu[\delta_\nu] + T_\lambda[\delta_\lambda] - s_\mu(\delta_\nu,\delta_\lambda)$.
If~$T_\nu[\delta_\nu] = \infty$ or~$T_\lambda[\delta_\lambda] = \infty$,  then this is trivial.  So assume~$T_\eta[\delta_\eta] < \infty$ for~$\eta \in \{\mu,\lambda\}$. 
Then, for~$\eta \in \{\nu,\lambda\}$ there exists an $\eta$-feasible set~$X_\eta \subseteq \rfaces{\eta}$ with $\Delta_\eta(X_\eta) = \delta_\eta$ and $|X_\eta| = T_\eta[\delta_\eta]$.  Since~$\delta_\nu$ and~$\delta_\lambda$ are compatible and since $X_\eta$ is $\eta$-feasible for $\eta \in \{\nu,\lambda\}$, $X_\nu \cup X_\lambda$ is $\mu$-feasible and satisfies~$\Delta_\mu(X_\nu \cup X_\lambda) = \delta_\mu$. 
Since~$T_\mu[\delta_\mu]$ is the smallest size of a set with this property, we obtain~$T_\mu[\delta_\mu] \le |X_\nu \cup X_\lambda| = |X_\nu| + |X_\lambda| - s_\mu(\delta_\nu,\delta_\lambda)$.  %

Now we show ``$\ge$'': 
Let~$\delta_\mu \in \Omega(\mu)$.  If~$T_\mu[\delta_\mu] = \infty$, then the statement holds trivially. So assume $T_\mu[\delta_\mu] < \infty$.  Then there exists a $\mu$-feasible set~$X_\mu \subseteq \rfaces{\mu}$ with~$\Delta_\mu(X_\mu) = \delta_\mu$ and~$|X_\mu| = T_\mu[\delta_\mu]$.  For~$\eta \in \{\nu,\lambda\}$ let~$X_\eta = X_\mu \cap \rfaces{\eta}$ and let~$\delta_\eta = \Delta_\eta(X_\eta)$.  Then, we have~$T_\eta[\delta_\eta] \le |X_\eta|$ for~$\eta \in \{\nu,\lambda\}$ and~$|X_\mu| = |X_\nu| + |X_\lambda| - s_\mu(\delta_\nu,\delta_\lambda)$.  We thus obtain~$T_\mu[\delta_\mu] \ge T_\nu[\delta_\nu] + T_\lambda[\delta_\lambda] - s_\mu(\delta_\nu,\delta_\lambda)$.  Since~$\delta_\nu$ and~$\delta_\lambda$ are compatible and~$\delta_\mu = \delta_\nu \oplus \delta_\lambda$, this pair is considered when taking the minimum in \eqref{eq:dp-table}. %
\end{claimproof}

Following the recurrence (\ref{eq:dp-table}), we can compute the entries of the tables $T_\mu$ from the tables $T_\nu$ and $T_\lambda$ in a bottom-up fashion.  We note that, by definition, the root~$\rho$ has no boundary vertices or faces and thus it has only a single descriptor, which we denote by~$\bot$.  The entry~$T_\rho[\bot]$ then stores the size of an optimal $\rho$-feasible set of faces, i.e., an optimal solution to \textsc{Non-Orthogonal Face Minimization$^*$}.  

For the running time, observe that each node~$\mu$ has at most~$w$ boundary vertices and faces and thus at most~$2^w \cdot 5^w = 10^w$ different descriptors. Hence, to compute $T_\mu[\delta]$, we need to consider at most~$10^{2w}=2^{2w\log_210}$  descriptor pairs, each of which can be evaluated in $O(w)$ time.
Since the tree $\cal T$ contains $O(n)$ nodes, the time to fill all the tables is thus~$2^{O(w)}n +\tau$, provided that a sphere-cut decomposition of $G$ of width at most $w$ can be computed in time~$\tau$.
\end{proof}

We use \cref{thm:fpt_sc} to derive results for graph classes where a sphere-cut decomposition can be computed efficiently. By \cref{th:sc-construction}, an optimal sphere-cut decomposition can be computed in cubic time for general plane graphs. By \cref{prop:sp-sc}, if the graph is series-parallel, a sphere-cut decomposition of width 2 can be computed in linear time. By \cref{prop:sc-kouterplane}, a sphere-cut decomposition of an $n$-vertex $k$-outerplane graph of $O(k)$ width can be computed in $O(kn)$ time.
We formalize these implications in the following corollaries.

\begin{corollary}\label{cor:fpt_tw}
For an $n$-vertex plane graph $G$, \textsc{Non-Orthogonal Face Minimization} can be solved in~$2^{O(\tw(G))}n + O(n^3)$ time.
\end{corollary}

\begin{corollary}\label{cor:series-parallel}
For plane series-parallel graphs, \textsc{Non-Orthogonal Face Minimization} can be solved in linear time.
\end{corollary}

\begin{corollary}\label{cor:fpt_outerplanar}
For $n$-vertex $k$-outerplane graphs, %
\textsc{Non-Orthogonal Face Minimization}$^*$ can be solved in $2^{O(k)}n$ time.
\end{corollary}

Based on \cref{cor:fpt_outerplanar}, we give an FPT algorithm for the decision version of the problem, parameterized by the natural parameter~$h$.

\begin{theorem}\label{thm:fpt_natural}
For $n$-vertex plane graphs, \textsc{$h$-Non-Orthogonal Faces} can be solved in~$2^{O(h)}n$ time.
\end{theorem}

\begin{proof}
	Let $G$ be a plane graph. Recall that in an optimum solution of \textsc{Non-Orthogonal Face Minimization}, a face may be selected to be non-orthogonal only if it is incident to at least one vertex of degree greater than four.  Let $V'$ be the set of vertices of $G$ that are incident to some face of $G$ that is incident to some vertex of degree greater than four. In particular, $V'$ contains all vertices of degree greater than four. Let $G'$ be the graph induced by the vertices in $V'$. By the above observation, it suffices to solve  \textsc{$h$-Non-Orthogonal Faces} for $G'$. %
    We claim that, if $G'$ is a yes-instance, then $G'$ is $k$-outerplanar with $k \leq 4h+2$.
    
    Let $F'$ be the set of non-orthogonal faces of $G'$, and consider an outer layering $L_1,\dots, L_k$ of $G'$. 
    Observe that each face of $F'$ only contains vertices of two consecutive layers $L_i$ and $L_{i+1}$, with $1 \leq i \leq k$. I.e., there are at most $2h$ layers that are covered by faces in~$F'$. Suppose, for a contradiction, that $k > 4h+2$. This implies that there exist at least three consecutive layers $L_j$, $L_{j+1}$, $L_{j+2}$ whose vertices are not incident to any face in~$F'$: Otherwise there can be at most $2h+2$ layers that do not contain vertices incident to some face in $F'$. Together with the at most $2h$ layers containing vertices incident to some face in~$F'$ this would be at most $4h+2$ layers.
    Let $v$ be a vertex of $L_{j+1}$. Since $v$ belongs to $G'$, it is incident to a face $f_v$ that is incident to a vertex $w$ (possibly, $w=v$) of degree greater than four. Note that $w \in L_j \cup L_{j+1} \cup L_{j+2}$. However, this implies that $w$ is a vertex of degree greater than four that is not incident to any face in $F'$, a contradiction.

	We proceed as follows. Construct $G'$ and test whether it is $(4h+2)$-outerplanar. If not, reject the instance. Otherwise, apply \cref{cor:fpt_outerplanar} to $G'$ and $V'=V(G')$.
\end{proof}

\section{Polynomial-Time Approximation Scheme for Fixed Embedding}\label{sec:ptas}

A polynomial-time approximation scheme (PTAS) is a series of algorithms $(\mathcal A_\epsilon)$ for an optimization problem such that for any $\epsilon > 0$ the algorithm $\mathcal A_\epsilon$ computes a feasible solution with relative error bound $\epsilon$. For the problem \textsc{Non-Orthogonal Face Minimization} this means that $\mathcal A_\epsilon$ computes a feasible solution $F'$, i.e., each vertex has residual degree at most~$4$ w.r.t.\ $F'$,
 whose size $|F'|$
is at most $(1+\epsilon)$ times the number of non-orthogonal faces in an optimum~solution.

\begin{theorem}\label{thm:ptas}
	\textsc{Non-Orthogonal Face Minimization} with fixed embedding admits~a~PTAS. 
\end{theorem}

\begin{figure}
		\centering
	\begin{minipage}[b]{0.32\linewidth}
		\centering
		\includegraphics[page=19]{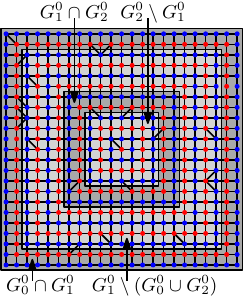}
		\subcaption{\nolinenumbers}
        \label{fig:ptas0}
	\end{minipage}\hfil
	\begin{minipage}[b]{0.32\linewidth}
	\centering
	\includegraphics[page=20]{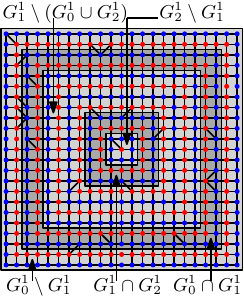}
	\subcaption{\nolinenumbers}
    \label{fig:ptas1}
\end{minipage}\hfil
	\begin{minipage}[b]{0.32\linewidth}
	\centering
	\includegraphics[page=21]{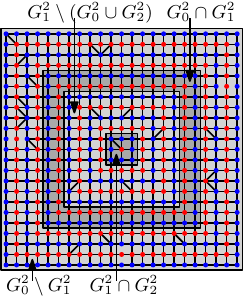}
	\subcaption{\nolinenumbers}
    \label{fig:ptas2}
\end{minipage}\hfil
	\caption{Blue and red vertices alternatingly show an outer layering into 12 layers. If $\epsilon =1/6$, we split the graph in three ways into subgraphs consisting of at most $2\lceil \frac{1}{\epsilon} \rceil + 2=8$ consecutive layers. Two consecutive graphs overlap in two layers (dark gray). 
    (a) $s=0$; (b) $s=1$; and (c) $s=2$.}
    \label{fig:ptas}
\end{figure}

\begin{proof}
	We use Baker's technique~\cite{baker94}, i.e., we try $1/\epsilon$ decompositions of a plane graph into overlapping subgraphs of bounded outerplanarity index, solve optimally on each, and argue that at least one decomposition has few faces in the optimum solution that cross the boundaries between pieces. This yields a ($1+\epsilon$)-approximation.  
    
    Let $G$ be a plane graph. First, we compute an outer layering. Recall that the vertices on the boundary of any inner face of $G$ are on at most two consecutive layers. 
	Let $k = 2 \cdot \lceil 1/\epsilon \rceil$. For each $0 \leq s < k/2$, we decompose $G$ as follows: The graph $G^s_i$, $i \geq 0$ is the graph induced by the $k+2$ layers indexed $2s + (i-1)k,\dots,2s + ik + 1$. Observe that $G_i$ is empty if $i > n/3k + 2$, since each layer contains at least three vertices. Moreover, $G^s_i$ and $G^s_{i+1}$ overlap in the two layers indexed $2s + ik$ and $2s + ik +1$. Each graph $G_i$ is at most $(k+2)$-outerplanar. Thus, by \cref{cor:fpt_outerplanar}, 
	the problem can be solved optimally in time $2^{O(k)}|V(G_i)|$. More precisely, we construct a minimum set $F^s_i$ of faces, such that all vertices in $L_i^s := L_{2s + (i-1)k + 1},\dots,L_{2s+ik}$ have residual degree at most four. Thus, $F_i^s$ is the outcome of \textsc{Non-Orthogonal Face Minimization$^*$} applied to the connected components of $G_i^s$ and~$L_i^s$. Observe that the connected components of $G_i^s$ share the outer face, however the vertices incident to the outer face of $G_i^s$ are not in $L_i^s$. So the outer face of $G^i_s$ is not contained in any optimal solution. 
    
    Let $F_s = \bigcup_{i}F_i^s$. 
    Then $F_s$ is a feasible solution for $G$: Let $v$ be a vertex of $G$ in layer $L_\ell$.
	Then there is an $i \geq 0$ and a $1 \leq j \leq k$  such that $\ell = s+(i-1)k + j$. It follows that $v$ is in~$L^s_i$ and 
	  $v$ is incident to at most four edges in   $E(G) \setminus E(F_s^i) \supseteq E(G) \setminus E(F_s)$.
	   
	Let $F_\textsc{opt}$ be an optimal solution for $G$ and let $F_\textsc{opt} \cap G^s_i$ be the set of faces in $F_\textsc{opt}$ that are only bounded by vertices and edges in $G^s_i$. Then, $F_\textsc{opt} \cap G^s_i$ is a feasible solution  for \textsc{Non-Orthogonal Face Minimization$^*$} applied to $G_i^s$ and~$L_i^s$, implying that $|F_\textsc{opt} \cap G^s_i| \geq |F^s_i|$. 
	
	For $0 \leq s < k/2$, let $F_\cap^s$ be the set of faces of $F_\textsc{opt}$ that are contained in two subgraphs $G^s_i$ and $G^s_{i+1}$ for some $i$, i.e., $F_\cap^s = \bigcup_{i \geq 0} \left((F_\textsc{opt} \cap G^s_i) \cap (F_\textsc{opt} \cap G^s_{i+1}) \right)$. Then $F_\cap^s$, $0 \leq s < k/2$ is a partition of $ F_\textsc{opt}$ into $k/2$ sets. So the smallest, say  $F_\cap^{z}$ has size at most $2/k \cdot  |F_\textsc{opt}|$. 
	We get, 
    
	\begin{equation*}   
	|F^z| \leq \sum_{i \geq 0}  |F^z_i| \leq \sum_{i \geq 0} |F_\textsc{opt} \cap G^z_i| = |F_\textsc{opt}| + |F_\cap^{z}| \leq (1+2/k) \cdot |F_\textsc{opt}| \leq (1+\epsilon) \cdot |F_\textsc{opt}|.
	\end{equation*} 
	
	Summarizing, for each $0 \leq s < \lceil 1/\epsilon \rceil$ and for each $0 \leq i \leq \frac{\epsilon}{6} \cdot n + 2$ compute optimum solutions $F^s_i$ for $G_i^s$ in $2^{O(1/\epsilon)}|V(G_i)|$ time. Among $F_s = \bigcup_{i \geq 0} F^s_i$, $0 \leq s < \lceil 1/\epsilon \rceil$ take the one with the fewest faces. This yields an $\epsilon$-approximation with run time $2^{O(1/\epsilon)}|V(G)|$.
\end{proof}

\section{Fixed-Parameter Tractability for Variable Embedding}

By using dynamic programming on the SPQR-tree and applying the techniques from the proof of \cref{thm:fpt_sc} to handle the R-nodes,  we show that \textsc{Non-Orthogonal Face Minimization} is FPT parameterized by the treewidth even if no embedding is given.  

\begin{theorem}\label{thm:fpt_tw_variable}
\textsc{Non-Orthogonal Face Minimization} can be solved for biconnected planar graphs without fixed embedding and with $n$ vertices and treewidth $k$ in time~$2^{O(k)}n + O(n^3)$.
\end{theorem}

\begin{proof}
    Let $G$ be a biconnected planar graph with $n$ vertices.
    We compute an SPQR-tree $\mathcal T$ of $G$ and root it at an arbitrary $Q$-node $\rho$. Let $\nu \neq \rho$ be a node of $\mathcal T$ and let $s$ and $t$ be the poles of $\nu$. 
    Given a planar embedding of $G(\nu)$, we consider the two $s$-$t$-paths of $G(\nu)$ on the outer face as the boundary of the \emph{left and the right outer face} of $G(\nu)$, i.e., we consider the outer face of $G(\nu)$ as if it was separated by parent$(\nu)$.      
    
    We use dynamic programming in order to compute the set $F'$ of non-orthogonal faces. We proceed bottom-up in $\mathcal T$. For each node $\nu$ of the SPQR tree, we fill in a table with at most 100 entries for the following possibilities: The residual degree of each pole is a number in $\{0,1,2,3,4\}$ and the left and the right outer face both have the choice to be or not to be in $F'$.

    For the Q-nodes, there are four possibilities for the choice of the two outer faces. The entry is $0$  if (a) the residual degree of both vertices is one and at least one of the outer faces is not in $F'$ or (b) the residual degree is zero for both vertices and both outer faces are in $F'$. Otherwise the entry is $\infty$.

    Let now $\nu$ be a node with the children $\nu_i$, $i=1,\dots,\ell$, let $e_i$ be the edges of skel$(\nu)$ that correspond to $\nu_i$ and let parent$(\nu_i)=(s_i,t_i)$. Assume first that $\nu$ is an S-node. We may assume that $e_1,\dots,e_\ell$ are in this order in skel$(\nu)$. Let $G^{(1)}(\nu)=G(\nu_1)$ and let $G^{(i)}(\nu)$, $i=2,\dots,\ell$ be the graph obtained from  $G^{(i-1)}(\nu)$ and $G(\nu_i)$ by merging $t_{i-1}$ and $s_i$. Then $G(\nu)=G^{(\ell)}(\nu)$. Iteratively, we compute records of $G^{(i)}(\nu)$, $i=2,\dots,\ell$ with poles $s_1$ and $t_i$. Observe that the desired record for $G^{(i)}(\nu)$ already determines the residual degree of $s_1$ in $G^{(i-1)}(\nu)$ and the residual of $t_i$ in $G(\nu_i)$ and the left and right outer face of both. Among all these entries, we combine the at most five entries for the residual degree $k$ of $t_{i-1}$ in $G^{(i-1)}(\nu)$ with the residual degrees at most $4-k$ for $s_i$ in $G(\nu_i)$ and take the one with the fewest inner faces for $F'$. The run time for one S-node $\nu$ is in $O(\deg \nu)$. It follows that the run time for all S-nodes is linear in the  number of edges of $\mathcal T$ and, thus, in $O(n)$.
    
	Assume now that $\nu$ is a $P$-node.  There can be at most four children where one pole has a residual degree greater than zero and at most four children where the other pole has a residual degree greater than zero. There are $\ell \choose 8$ possible choices for these eight components; all other components must have residual degree zero, which implies that the outer faces must both be in $F'$. For the chosen eight components, we check all possible combinations that yield the desired entry for $\nu$. This creates in total at most 7 new inner faces. Each of these new inner faces is in $F'$ if the respective outer face of one of the adjacent components was in~$F'$. We add their number to the sum of the number of inner faces in the respective records for all $\ell$ child components. If $\ell > 8$, we add the remaining components into one of the 9 faces created by the 8 chosen components. If we add them into a face that was not in $F'$ then the number of inner faces in $F'$ increases by $\ell-8$, otherwise by $\ell-9$. Per choice of the eight components this can be done in constant time. In the end, we maintain for each entry of $\nu$ one with the fewest inner faces in $F'$. The total run time for one $P$-node $\nu$ is  $O(\deg^8 \nu)$. 

    In order to speed up the run time, consider for each child $\nu_i$ the number $k_i^{(0)}$ of inner faces needed to obtain zero residual degree for both poles. More generally, for each choice~$R$ of the residual degrees for each of the two poles and for each of the two outer faces, let $k_i^{(R)}$ be the number of inner faces needed to fulfill $R$ for $G(\nu_i)$. 
    Let  $k^{(0)} = \sum_{i=1}^\ell k_i^{(0)}$. 
    Assume now that we picked the eight children $\nu_{i_j}$, $j=1,\dots,8$ and for child $\nu_{i_j}$, we considered the choice~$R_j$. Then the total number of faces in $F'$ is $k^{(0)} - \sum_{j=1}^8(k_{i_j}^{(0)}-k_{i_j}^{(R_j)})$ plus the number of faces in $F'$ between any two components. But the latter only depends on the choice for the outer faces of the component, i.e., on the choice of $R_{j}$, $j=1,\dots,8$,  and on the ordering of the components, but not on the choice of the eight components.  Thus,
    we store for each record~$R$ a set of eight children for which $k_i^{(0)}-k_i^{(R)}$ is largest.  Now we pick the 8 components with residual degree possibly different from zero for the poles only among these at most 800 children. This results in a constant number of choices. 
    Thus, each $P$-node $\nu$ can be handled in  $O(\deg \nu)$ time and all $P$-nodes together in $O(n)$ time.

    For an R-node $\nu$,  we adopt the approach in \cref{thm:fpt_sc}, i.e., we compute a sphere-cut decomposition of skel$(\nu)$ without the edge parent$(\nu)$. Observe that the treewidth of skel$(\nu)$ is bounded by tw$(G)$. When applying the dynamic program described in the proof of  \cref{thm:fpt_sc}, the records for the leaves of the sphere-cut-decomposition, i.e., the edges $e_i$, $i=1,\dots,\ell$ of skel$(\nu)$ are the records of $G(\nu_i)$. In the first place, the algorithm will compute for each combination of residual degrees for the poles an entry $R_L$ for the case where the left outer face is in $F'$ and the right outer face is not and an entry $R_R$ where the right outer face is in $F'$ and the left outer face is not. Since the embedding of skel$(\nu)$ is only fixed up to a flip, we take the minimum value among $R_R$ and $R_L$ for both cases.
    Per $R$-node $\nu$, all entries can be computed in $2^{O(\tw(G))}\deg \nu + O(\deg^3 \nu)$ time and all $R$-nodes together in $2^{O(\tw(G))}n + O(n^3)$~time. 
    
    The root of $\mathcal T$ is handled as a $P$-node with two children: one is its actual child and the other is a $Q$-node.	
	Once we have computed the entries for the root, we consider only those with a finite value, where either both outer faces are in $F'$ (in which case we add one to the number of faces in $F'$) or neither is in $F'$ and take the one with the fewest faces in $F'$. 
\end{proof}

Since all nodes $\nu$ other than the $R$-nodes could be handled in deg$(\nu)$ time
in the proof of \cref{thm:fpt_tw_variable}, 
we obtain the following corollary.

\begin{corollary}
\textsc{Non-Orthogonal Face Minimization} can be solved in linear time for biconnected series-parallel graphs without fixed embedding.
\end{corollary}

\section{Conclusion}
We studied the problem of computing planar drawings with a minimum number of non-orthogonal faces. We showed that the problem is \NP-hard even for $3$-connected graphs.  For plane graphs, we gave FPT algorithms and a polynomial-time approximation scheme. In the variable-embedding setting, we gave an FPT algorithm for biconnected graphs parameterized by treewidth. We wonder whether the last result can be extended to connected planar graphs.

All our algorithms have a singly-exponential dependence on the parameter.  It is an interesting question whether also sub-exponential FPT algorithms exist and for which parameterizations a polynomial-size kernel can be obtained.

\bibliography{references}

\clearpage

\appendix

\section{Omitted Proofs}

\characterization* \label{lem:char*}

\begin{proof}
	Consider first a planar drawing $\Gamma$ of $G$ and let $F'$ be the set of non-orthogonal faces in $\Gamma$. Let $v$ be a vertex of $G$ and let $e$ be an edge that is incident to $v$. If $e$ is a residual edge w.r.t.\ $F'$, i.e., if  one of the two faces incident to $e$ is not in $F'$ then $e$ is incident to an orthogonal face and, hence, must be drawn as an orthogonal polyline. But at most four edges incident to $v$ can be drawn as an orthogonal polyline, two starting with a vertical segment at $v$ and two starting with a horizontal segment at $v$.  It follows that the residual degree of $v$ w.r.t.\ $F'$ is at most four.
	
	Consider now a set $F'$  of at most $h$ faces 
	of $G=(V,E)$ (in some/the given embedding) such that each vertex has residual degree at most four w.r.t.\ $F'$. Let $G'=(V, E \setminus E(F'))$ %
    be the spanning subgraph of $G$ containing exactly the residual edges. $G'$ has maximum degree four and, thus, admits an orthogonal drawing $\Gamma'$. We may assume that $\Gamma'$ respects the embedding of $G$, in particular also  the relative position of the connected components of $G'$.  We now add the remaining edges, i.e., the edges in $E(F')$ as some polyline between the end vertices into $\Gamma'$ such that the chosen/given embedding is respected. Since the edges in $E(F')$ are only incident to faces in $F'$ it follows that at most the faces in $F'$ are non-orthogonal. Thus, $\Gamma'$ is $h$-near orthogonal.
\end{proof}

\scgeneral* \label{th:sc-construction*}

\begin{proof}
    We have $\bw(G) \le \tw(G)+1$.  Using a result of Biedl and Vel\'azquez~\cite[Lemma 1]{BiedlV13}, we can augment $G$ by adding edges to a triangulation~$G'$ with~$\tw(G') = \max \{3,\tw(G)\}$.  We then compute a sphere-cut decomposition $\cal T'$ of $G'$ of width $\bw(G')$ in $O(n^3)$ time using the algorithm from~\cite{DornPBF10}.  Observe that~$\bw(G') \le \tw(G')+1 \le \max\{3,\tw(G)\}+1 \le \max\{3,\lfloor 3/2 \bw(G)\rfloor\} + 1 \in O(\bw(G))$.  
    By removing the triangulation edges, we obtain a sphere-cut decomposition $\cal T$ of $G$ of no greater width from a sphere-cut decomposition $\cal T'$ of $G'$. 
    Specifically, $\cal T$ can be obtained by applying the following procedure.
    First, we initialize $\cal T = \cal T'$. 
    We replace an appearance of a face $f'$ in any combinatorial noose by the face of $G$ containing $f'$.
    Observe that the leaves of $\cal T$ that correspond to the removed triangulation edges still belong to $\cal T$, however the noose associated to each such node encloses neither an edge nor a vertex in its interior. As long as there exists a leaf $\mu$ in $\cal T$ whose noose has this property, i.e., it does not enclose any edge or vertex in its interior, we perform the following steps.
    First, remove $\mu$ from $\cal T$. 
    Second, let $\nu$ be the neighbor of $\mu$ and let $\tau_1$ and $\tau_2$ be the neighbors of $\nu$ different from $\mu$; remove $\nu$ from $\cal T$ and introduce the  edge $e$ connecting $\tau_1$ and $\tau_2$. Observe that the two edges $e_i$, $i=1,2$,  between $\nu$ and $\tau_i$ corresponds to the same partition $E_1,E_2$ of the edges of $G$. In order to construct the noose for the new edge $e$, we take the noose $\gamma$ of $e_1$. For each vertex $v$ that is only incident to edges in $E_1$ or only to edges in $E_2$, we remove $v$ and its predecessor face in $\gamma$ from $\gamma$%
    .
    Note that, after each application of this procedure, $\cal T$ is a ternary tree whose nodes are a subset of the nodes of $\cal T'$ (and whose nooses are, in particular, subsequences of the nooses in $\cal T'$). Hence, eventually, this procedure yields a sphere-cut decomposition of $G$ of width~$O(\bw(G))$.
\end{proof}

\scseriesparallel*\label{prop:sp-sc*}

\begin{proof}
Let $G$ be a biconnected plane series-parallel graph and let $\cal T$ be its SPQR-tree. We turn~$\mathcal T$ into a sphere-cut decomposition by replacing $S$-nodes and $P$-nodes as follows. We replace each S-node~$\mu$ with more than two children by a rooted ternary tree $T_\mu$ whose leaves correspond to the virtual edges of~$\skel(\mu)$, in such a way that the leaves of each subtree of $T_\mu$ correspond to consecutive virtual edges in~$\skel(\mu)$.  Similarly, we replace each P-node~$\mu$ with more than two children by a rooted ternary tree $T_\mu$ whose leaves correspond to the virtual edges of~$\skel(\mu)$, in such a way that the leaves of each subtree of $T_\mu$ correspond to consecutive virtual edges in the embedding of~$\skel(\mu)$ determined by the embedding of $G$. This can be done in linear time.  This modified SPQR-tree~$\cal T'$ is a branch decomposition and in fact a sphere-cut decomposition of $G$ of width $2$, where the noose of each edge $e$ of $\cal T'$ is precisely the two poles shared by the skeletons of the endpoints of~$e$.
\end{proof}

\end{document}